\documentclass{amsart}
\usepackage[english]{babel}
\usepackage[utf8]{inputenc}
\usepackage{bbm}
\usepackage{verbatim}

\newcommand{\N}{\ensuremath{\mathbb{N}}}
\newcommand{\R}{\ensuremath{\mathbb{R}}}

\newcommand{\E}{\ensuremath{\mathbb{E}}}
\renewcommand{\P}{\ensuremath{\mathbb{P}}}

\newcommand{\ind}[1]{\ensuremath{\mathbbm{1}_{\left\{#1\right\}}}}
\newcommand{\cal}[1]{\ensuremath{\mathcal{#1}}}

\def\var{\mathrm{var}}
\newcommand{\diff}{\mathop{}\mathopen{}\mathrm{d}}

\newtheorem{definition}{Definition}
\newtheorem{proposition}{Proposition}
\newtheorem{corollary}{Corollary}[proposition]

\newtheorem{lemma}{Lemma}
\newtheorem{theorem}{Theorem}

\title{A stochastic model of the production of multiple  proteins in cells}

\date{\today}

\author[V. Fromion]{Vincent Fromion}
\email{Vincent.Fromion@jouy.inra.fr}
\address[V. Fromion, E. Leoncini]{MIG Mathematic, Computing Science and Genome, INRA,
Domaine de Vilvert, 78350 Jouy-en-Josas, France}
\author[E. Leoncini]{Emanuele Leoncini}
\email{Emanuele.Leoncini@inria.fr}
\author[Ph. Robert]{Philippe Robert}
\email{Philippe.Robert@inria.fr}
\address[E. Leoncini,Ph. Robert]{INRIA Paris---Rocquencourt, Domaine de Voluceau, 78153 Le Chesnay, France}

\keywords{}

\date{\today}

\begin{document}

\begin{abstract}
The production processes of proteins in prokaryotic cells are investigated.  Most of the mathematical models in the literature study the production of {\em one} fixed type of proteins. When several classes of proteins are considered, an important additional aspect has to be taken into account, the limited common  resources of the cell (polymerases and ribosomes) used by the production process.  Understanding  the impact of this limitation is a key issue in this domain. In this paper we focus on the allocation of ribosomes in the case of the production of multiple proteins.  The cytoplasm of the cell being a disorganized medium subject to thermal noise,  the protein production process has an important stochastic component. For this reason, a Markovian model of this process is introduced. Asymptotic results of the equilibrium are obtained under a scaling procedure and a realistic biological assumption of saturation of the ribosomes available in the cell.  It is shown in particular that,  in the limit,  the number of non-allocated ribosomes at equilibrium converges in distribution to a Poisson distribution whose parameter satisfies a fixed point equation. It is also shown that the production process of different types of proteins can be seen as  independent production processes but with modified parameters. 
\end{abstract}

\maketitle 

\bigskip

\hrule

\vspace{-3mm}

\tableofcontents

\vspace{-1cm}

\hrule

\bigskip

\section{Introduction}
 The \emph{gene expression} is the process by which the genetic information of a biologic cell is synthesized into a functional product, the proteins.  The production of proteins is the most important cellular activity, both for the functional role and the high associated cost in terms of resources (in prokaryotic cells it can reach up to $85$\% of the cellular resources). In particular, in a B. subtilis  bacterium there are about $3.6\cdot 10^6$ proteins of approximately $2000$ different types with a large variability in
concentration, depending on their types: from a few dozen up to $10^5$. See e.g. Muntel et al.~\cite{Muntel}  where the absolute quantification of cytosolic  proteins  is reported in different  growth rate conditions.

The information flow from DNA genes to proteins is a {\em fundamental process}, common to all living organisms and is composed of two main fundamental processes: \emph{transcription} and \emph{translation}, described briefly thereafter. 
\begin{itemize}
\item {\em Transcription}.  During this step, the {\em RNA polymerase} binds to the DNA and makes a copy, an mRNA, of a specific DNA sequence. The mRNA is a chemical “blueprint” for a specific protein and is constituted of nucleotides.
\item {\em Translation}. The assembly of a polypeptide chain is achieved through a large complex: the {\em ribosome}. During this step, the ribosome binds to an mRNA and builds the polypeptide chain using the mRNA as a template.

\end{itemize}
See Chapter~12 and~14 of Watson et al.~\cite{Watson} for a much more detailed description of these complex processes. 
\subsection*{Stochastic Nature of Transcription and Translation}
The gene expression is a highly stochastic process: protein production is the result of the realization of a very large number of elementary stochastic processes.  The randomness is basically due to the thermal excitation of the environment. It drives the free diffusion of the main components, polymerases, messengers and ribosomes in the cytoplasm which can be seen as a viscous fluid.  For this reason it impacts the pairing of two cellular components freely diffusing through the cytoplasm or it can cause the spontaneous rupture of the bonds and subsequent release of paired components such as polymerases on DNA or ribosomes on mRNAs before one of the steps transcription or translation is completed.

It should be noted that this description is largely simplified since the active processes associated to polymerases and ribosomes are composed of highly sophisticated steps.  For example, once a polymerase is bound to the gene, the messenger chain is built through a series of specific stochastic processes, in which the polymerase recruits one of the four nucleotides in accordance to the DNA template.  Additionally, a dedicated proof reading mechanism takes place during this process.  A similar description is associated to the translation step.  

\subsection*{Competition for Resources of the Cell}
Each type of protein has its dedicated gene in the DNA and its associated mRNAs but the polymerases and ribosomes are the common resources used in the productions process. Due to their complexity, the cost of these  complex molecules  is high and, for this reason, their number is controlled within the cell.  In first approximation, protein production can be seen as a competition between genes [resp. mRNAs]  for polymerases [resp. ribosomes].  

\subsection*{Literature}
The first stochastic models of the production of {\em one} type of protein were investigated in the late 70s by Berg~\cite{Berg1978} and Rigney~\cite{Rigney1979, Rigney1977}, Shahrezaei and Swain~\cite{Swain2008} and reviewed by Paulsson~\cite{Paulsson}.  For a long period of time, it has not been possible to compare these theoretical results to real data, because of the lack of appropriate laboratory techniques. In the last two decades, the introduction of reliable expression reporter techniques and the use of fluorescent reporters, such as the GFP (Green Fluorescent Protein), has allowed observations in live cells and the experimental quantification of the protein production at single cell level. See Taniguchi et al.~\cite{Taniguchi} for the experimental characterization of a large number of messengers and proteins of E. Coli.  See also Fromion et al.~\cite{Fromion}.  In these studies, only one type of protein is considered, the variance at equilibrium of the concentration of the protein is estimated in terms of the rates of gene activation/inactivation, the rates of binding of polymerases [resp. ribosomes] to genes [resp. mRNAs] and the degradation rates. The main goal of these works is to obtain a better understanding of the mechanisms which allow the cell to produce the right proportions of proteins despite the intrinsic randomness and to control to some extent the variance of the protein production system.

Concerning the case of multiple proteins, the mathematical results concerning the allocation of polymerases/ribosomes to the production of multiple proteins seem to be quite scarce. The paper Mather et al.~\cite{Mather} discusses the impact of the fluctuations of the number of mRNAs on the number of proteins of each type at equilibrium. In the case of two types of proteins  formulas for mean and variances  at equilibrium are provided.  

\subsection*{A Model of Protein Production}
We describe the general model used to analyze the production of multiple types of proteins. We focus on the competition for ribosomes, in particular the number of mRNA's associated to each type of protein is supposed to be fixed. 

\medskip
\noindent
{\bf Model for One Type of  Protein.}\\
We first recall the classical model for the production of one type of protein, see Paulsson~\cite{Paulsson} for example. 
\begin{enumerate}
\item A ribosome is bound on a given mRNA at rate $\lambda$ if the total number of ribosomes already bound to this mRNA is strictly less that $C>0$.
The quantity $\lambda$ represents the {\em affinity} of an mRNA associated to this protein.
\item A bound ribosome produces a protein at rate $\mu$.
\item Dilution/Proteolysis: we consider an effective rate $\gamma$ of protein decay, which accounts for both dilution and proteolysis.
\end{enumerate}
Note that in this case, there is no real restriction on the number  of available ribosomes, these resources are, in some way, represented by the affinity parameter $\lambda$ as a global variable. If the number of mRNA's is fixed,  the number $(X(t))$ of ribosomes bound to a given mRNA is in fact an $M/M/C$ queue, see Asmussen~\cite{Asmussen}. In particular  the distribution of $(X(t))$ at equilibrium is a truncated geometric distribution $(\rho^k/Z, 0\leq k\leq C)$ with  $\rho=\lambda/\mu$ and $Z$ is the normalization constant. 

Up to now, most of explicit formulas for the distribution of the number of proteins have been obtained in this simplified framework. 

\medskip
\noindent
{\bf A Multi-Type Model.}\\
 When considering the production of many types of proteins, we can divide the proteins according to their concentration.  If $P$ is the number of different concentrations of proteins present in the cell, for $1\leq p\leq P$,  we denote by $K_p$ the number of types of proteins with the concentration. For $1\leq k\leq K_p$, one will refer to as a type/class $(p,k)$ protein. Note that to each class of proteins corresponds a class of messenger RNAs.  It is assumed that a class $(p,k)$ mRNA can accept simultaneously at most $C_p$ ribosomes.  

The total number  of available ribosomes is $N$  and the  number of different types of  mRNAs $(K_p, 1\leq p\leq P)$  is fixed.  Ribosomes which are not bound to an mRNA are called {\em free ribosomes.}  
\begin{enumerate}
\item A free ribosome  binds  to an mRNA of class $(p,k)$  at rate  $\lambda_p$ provided that there are  less than $C_p$ ribosomes already attached on this mRNA.
\item In the case in which at least one ribosome is attached to a mRNA of class $(p, k)$, at
rate $\mu_p$ a protein of class $p$ is produced, a ribosome detaches from this mRNA and it
reintegrates the free ribosomes pool.
\end{enumerate}
The important assumption of the model is that the characteristics of the various proteins and mRNAs of class $(p,k)$ depend only on the parameter $p$ related to the concentration.  Without loss of generality, we assume that to each class $(p, k)$ corresponds one mRNA. This assumption is not restrictive since the quantity $K_p$ can be increased to take into account the case of several mRNAs per class $(p, k)$. 

It should be noted that,  since the number of ribosomes is assumed to be fixed,  the feedback effect due to the consumption of proteins to make new ribosomes is  not  taken into account in the analysis. This phenomenon may not be negligible as some studies suggest, see Warner et al.~\cite{Warner}. 

The stochastic model in Mather et al.~\cite{Mather} also assumes a fixed number of mRNAs and ribosomes.  The main difference 
with our description is that the authors encapsulate the number of  ribosomes attached to mRNAs as a  one dimensional  process with finite state space ${\cal S}$. To each state of ${\cal S}$ corresponds a rate at which a ribosome is attached to an mRNA of a given type or at which a ribosome is detached (and therefore a protein is created). In particular, there is no limitation on the number of ribosomes attached to a given mRNA, except that it must be less than the total number of ribosomes.  As it will be seen our description is more detailed since we keep the number of ribosomes attached to each mRNA to represent the production process.  As a consequence one deals with a (large) multi-dimensional process instead of a one-dimensional  process to describe the activity of ribosomes.  

\subsection*{Saturation Condition}
To stick with the  biological reality, it is assumed that a {\em saturation condition} holds, i.e. that the total  number $N$ of ribosomes  is  smaller than the number  that mRNAs could use in an unconstrained system.   Ribosomes are made of a large number of mRNAs and proteins and, for this reason,  are expensive components of the cell,  their number is therefore limited. See Warner et al.~\cite{Warner} in the case of yeast for example.  In this paper the saturation condition used is that  $N$ is small than the {\em average} number of ribosomes of an {\em unconstrained} system {\em at equilibrium}.  This notion is precised in Section~\ref{Equi}.   In Mather et al.~\cite{Mather} the saturation condition is that $N$ is smaller than the {\em maximum} number of ribosomes which can be attached to mRNAs. As it can be seen, with this later condition and appropriate parameters, the restriction on the number of ribosomes can be marginal since the number of ribosomes required at equilibrium  can be in some cases much smaller than the maximum number of ribosomes which can be attached to  mRNAs.

\subsection*{Mathematical Background: Limit Theorems of Gibbs Distributions}
As it will be seen, the equilibrium of the Markov process associated to this biological system can be expressed in terms of i.i.d. random variables conditioned on some event.  Equilibrium has a  product form distribution,  a Gibbs distribution in the language of statistical mechanics.  A classical scaling setting consists in letting the number of components (sites) go to infinity and to analyze the evolution of the distribution of the state of a component.  This situation has already been investigated in some queueing systems, see Kelly~\cite{Kelly:14}, see also~\cite{Anselmi,Fayolle,Malyshev} for example.  One can nevertheless trace these methods back to Khinchin~\cite{Khinchin} in the more general context of statistical mechanics  to obtain limit results in the thermodynamic limit, i.e. when the volume (number of particles) is going to infinity.  It is shown than local limit theorem can be used to prove such results.  See Chapter~V of Khinchin~\cite{Khinchin}.   Local limit theorems are refinements of the classical central limit theorem and therefore may be difficult to establish.  Dobrushin and Tirozzi~\cite{Dobrushin} simplified this approach by showing that an integral version of such a result is enough. 

\subsection*{Outline of the Paper}
The main results of this paper, Theorems~\ref{theo1} and~\ref{TheoEqui} consider the case where  the total number of ribosomes $N$ is large and when the number $K_p$ of mRNA's associated to proteins of class $1\leq p\leq P$ scales with respect to $N$.  At equilibrium, when $N$ gets large and under a saturation condition (ribosomes are a scarce resource),  it is shown that the number of free ribosomes converges in distribution to a  Poisson distribution with parameter $\delta<1$, where $\delta$ is the solution of a fixed point equation.   Furthermore,  the number of ribosomes bound to an mRNA of class $p$ converges in distribution to a truncated geometric distribution with parameter $\lambda_p\delta/\mu_p$.  

The main consequence of these results is the following: under this scaling regime and the saturation condition, the number of free ribosomes is converging in distribution to a finite random variable (recall that the total number $N$ of ribosomes goes to infinity). Additionally, the production of individual type of proteins become independent in the limit, meaning that the analysis of a specific type of protein of class p can still be performed using a classical one protein model, where the affinity parameter $\lambda_p$ which controls the translation speed is replaced by $\lambda_p\delta$. The parameter $\delta$ encodes therefore the interaction of the different production processes occurring simultaneously, through a deterministic fixed point equation.  The quantity $\delta$ can be seen as encoding in some way the interaction of the different production processes but only through a deterministic fixed point equation. It is also shown that for $n\geq 1$, the numbers of ribosomes bound to mRNA's of class $1\leq p_1\leq p_2\leq \cdots\leq p_n\leq P$ respectively become independent random variables in the limit as $N$ goes to infinity.  This is typical of a mean-field interaction property. See Sznitman~\cite{Sznitman:06} for example.

The paper is organized as follows. Section~\ref{secMod} introduces the Markov process describing the protein production steps. An explicit expression of its invariant distribution is derived. In Section~\ref{Equi}, the asymptotic results concerning the distribution of the number of free ribosomes are obtained as well as the independence property of the production processes of the different types of proteins.  The basic ingredients of the proof are a change of probability distribution and a strengthened version of the central limit theorem for lattice distributions, known as Edgeworth expansion.  The method turns out to be a probabilistic analogue of a saddle point method. One concludes in Section~\ref{BioSec} with a discussion of the significance of these results, in particular of the fixed point equation, from a biological point of view. Several simple expansions are derived. 

\subsection*{Acknowledgments} The third author is grateful to Danielle Tibi  for pointing out Reference~\cite{Dobrushin}.

\section{Stochastic Model}\label{secMod}
The notations and assumptions for the model of  protein  production are introduced.  The ribosome pool is supposed to be constituted of $N$ ribosomes.

The concentration of proteins have $P$ possible values. For $1\leq p\leq P$, there are $K_p$ different types of mRNAs associated with the  $p$th concentration. For $1\leq k\leq K_p$, the corresponding mRNA (resp. proteins produced by this mRNA) will be referred to as being of class $(p,k)$.  The cell described by this model has therefore $K_1+\cdots+K_P$ different types of proteins.  As noted before the fact that there is one mRNA per each $(p,k)$ class is not restrictive.

The affinity of mRNAs of class $(p,k)$, i.e. the rate at which a free ribosome binds on such mRNA, is denoted as $\lambda_p$.  Similarly, once a ribosome is bound to an mRNA of class $(p,k)$, a protein of class $(p,k)$ is produced at rate $\mu_p$. One defines $\rho_p=\lambda_p/\mu_p$.  A protein of class $(p,k)$ dies, via enzymatic degradation or dilution, at rate $\gamma_p$. Finally, the maximum number of ribosomes that can attach to an mRNA of class $(p,k)$ is denoted by $C_p$.

For $1\leq p\leq P$ and $1\leq k\leq K_p$, one denotes by  $X_{p,k}(t)$  the number of ribosomes attached to the mRNA of class $(p,k)$ at time $t$.  Note that
\[
N- \sum_{p=1}^P \sum_{k=1}^{K_p} X_{p,k}(t),
\]
is the number of free ribosomes at time $t$, i.e. the ribosomes which are not attached to any mRNA. 

\subsection*{State Space and $Q$-matrix for Ribosomes}
The Markov process describing the production  of proteins is    $(X(t))=((X_{p,k}(t),1\leq k\leq K_p),1\leq p\leq P)$, it is  an irreducible Markov process with values in the state space
\[
{\cal S}=\left\{\left(x_{p,k},\substack{1\leq k\leq K_p,\\1\leq p\leq P}\right): x_{p,k}\in \N,  0\leq x_{p,k}\leq C_{p}  \text{ and } \sum_{p=1}^P\sum_{k=1}^{K_p}x_{p,k}\leq N\right\}.
\]
Because of our assumptions, the $Q$-matrix $Q=(q(x,y), x, y\in{\cal S})$  is given by, for $1\leq k\leq K$ and $x=(x_{p,k})\in{\cal S}$, 
\begin{equation}\label{QMat}
\begin{cases}
q(x,x+e_{p,k})= \lambda_pr(x), &\text{ if } x_p+1\leq C_p ,\\
q(x,x-e_{p,k})= \mu_p, &\text{ if }x_p>0,
\end{cases}
\end{equation}
where 
\[
r(x)=N-\sum_{p=1}^P \sum_{k=1}^{K_p} x_{p,k}
\]
and $e_{p,k}$ is the unit vector associated to the coordinate $(p,k)$. The first equation of Relation~\eqref{QMat}  expresses the fact that any free ribosome binds at rate $\lambda_p$ on  the mRNA of type $(p,k)$ provided there is some room left. 

\subsection*{Variables for Proteins}
For $1\leq p\leq P$ and $1\leq k\leq K_p$, one defines as $Y_{p,k}(t)$  the number of proteins of class $(p,k)$ at time $t$. The process $(U(t))$ defined by
\[
(U(t))= (((X_{p,k}(t),Y_{p,k}(t)), 1\leq k\leq K_p), 1\leq p\leq P)
\]
is then a Markov process. If $U(t)=((x_{p,k},y_{p,k}))$, the rate at which the coordinate $(p,k)$ is increased by $(-1,1)$ is $\mu_p$ if $x_{p,k}>0$ (a ribosome is detached and a protein is created). Similarly, the coordinate is decreased by $(0,1)$ at rate $\gamma_py_{p,k}$. The jumps concerning the $(x_{p,k})$-coordinates only are the same as before. 

This class of models is  related to stochastic models of communication networks like loss networks.  See Kelly~\cite{Kelly:04} for example. The difference here is that the analogue of the ``input rate'', i.e. the rate at which free ribosomes binds to mRNAs, depends on the state of the process through the variable $r(x)$.  The invariant distribution has also a product form expression involving  Poisson random variables  instead of geometric random variables.  Another class of related models are the Gordon-Newel networks where a fixed number of customers/ribosomes travel through the different nodes of the network with some routing mechanism.  See Gordon and Newel~\cite{Gordon}. In a quite different biological context, a similar analogy with queueing networks has been used in Mather et al.~\cite{Mather2}. 

The following proposition shows that the Markov process describing the state of the ribosomes has an explicit invariant distribution together with a reversibility property.  See Kelly~\cite{Kelly} for a general presentation of such stochastic models.
\begin{proposition}\label{equi-prop}
The Markov process $(X(t))$ is reversible and its invariant distribution $\pi$ is given by, for $x=(x_{p,k})\in{\cal S}$, 
\begin{equation}\label{pi}
\pi(x)=\frac{1}{Z} \frac{1}{r(x)!}\prod_{p=1}^P\prod_{k=1}^{K_p} \rho_k^{x_{p,k} }
\end{equation}
where, for $1\leq p\leq P$,  $\rho_p=\lambda_p/\mu_p$, 
\[
r(x)=N- \sum_{p=1}^P \sum_{k=1}^{K_p} x_{p,k}
\]
 and $Z$ is the normalization constant. 
\end{proposition}
When  $(X(t))$ is in state $x=(x_{p,k})$,   the quantity
\[
\mu_p (x_{p,1}+x_{p,2}+\cdots+x_{p,K_p})
\]
is the total rate of production of proteins with the $p$th concentration and 
$r(x)$ is the number of free ribosomes. 
\begin{proof}
The proof is fairly simple, one has only to check, see Kelly~\cite{Kelly} for example, that, for any $1\leq p\leq P$ and $1\leq k\leq K_p$, the relation 
\begin{align*}
\pi(x)q(x,x+e_{p,k})&=\pi(x+e_{p,k})q(x+e_{p,k},x), \\
\intertext{ holds, i.e. that } 
\pi(x)\lambda_pr(x)
&=\pi(x+e_{p,k})\mu_p,
\end{align*}
for any $x\in {\cal S}$ such that $x+e_{p,k}\in {\cal S}$. The verification is straightforward. 
\end{proof}

The following corollary is a direct consequence of the above proposition, it expresses the invariant distribution of the number of free ribosomes in terms of geometric random variables. This representation will be useful to derive asymptotic results in Section~\ref{Equi}.
\begin{corollary}\label{propR}[Distribution of the Number of Free Ribosomes at Equilibrium]
If, for $1\leq p\leq P$ and $1\leq k\leq K_p$, $(G_{p,k})$ is an i.i.d. sequence of geometric random variables with parameter $\rho_p$ on the state space $\{0,\ldots,C_p\}$, then the distribution of the number $R$ of free ribosomes at equilibrium is given by, for $0\leq n\leq N$,
\begin{equation}\label{Free}
\P(R=n)=\frac{1}{Z_R} \frac{1}{n!}\P\left(\sum_{p=1}^P\sum_{k=1}^{K_p}G_{p,k}=N-n\right),
\end{equation}
where $Z_R$ is the normalization constant. 
\end{corollary}
The expression of the invariant distribution seems satisfactory at first sight since it gives an explicit representation of the equilibrium of the system.  However, because of the large size of the state space ${\cal S}$, it is hard to estimate the partition function $Z$ and, consequently, the same is true for interesting characteristics  of the system such as the average number of a given type of protein.
  In the following  a scaling approach to this model is considered, it consists in letting the number of free ribosomes $N$ and the number of classes of proteins go to infinity. As a consequence,  a qualitative description of the system is then possible, in particular to derive the rate of production of proteins.

This is also a common  property  shared with some stochastic models of  communication networks, known as loss networks, where the explicit expression of the invariant distribution is known but of little use in practice  for  networks with a reasonable number of nodes.   In this setting, a fruitful approach is to consider a scaling of the system by speeding up the input parameters by a large factor $N$ and analyze a rescaled version of the state of the network when $N$ is large.  See Kelly~\cite{Kelly:04}.

A scaling approach is also used in the following by taking advantage of the fact
that a large number of different types of proteins is produced in any condition within
cells. 
\section{Asymptotic Behavior of Equilibrium}\label{Equi}
From now on,  it is assumed that the number $N$ of ribosomes is large and that the number of classes of proteins with a fixed concentration $p$, $1\leq p\leq P$, scales with $N$ in the following way, 
\begin{equation}\label{beta}
K_p^N=N\beta_p+o\left(\sqrt{N}\right),
\end{equation}
in particular the sequence $(K_p^N/N)$ converges to $\beta_p$. 

An index N is added to the various quantities related to the production of proteins when in the system are present N ribosomes. Note that for $1\leq p \leq P$, $1 \leq j \leq K_p^N$ and any $N$, the number of ribosomes that can be attached to an mRNA of class $(p,k)$ is still bounded by the quantity $C_p$.

Proposition~\ref{equi-prop} shows that, at equilibrium,  the state of the process can be expressed  in terms of  independent geometric random variables $(G_{p,k})$   with parameter $\rho_p$ restricted to the state space $\{0,1,\ldots, C_p\}$ and a Poisson random variable $P_1$ with parameter $1$ so that,
\begin{multline}
\left(\left(X_{p,k}(t), 1\leq k\leq K_p^N\right), 1\leq p\leq P\right)\\ \stackrel{\text{def.}}{=}
\left(\left(\left(G_{p,k}, 1\leq k\leq K_p^N\right), 1\leq p\leq P\right)\left| P_1+\sum_{p=1}^P\sum_{k=1}^{K_p^N} G_{p,k}=N\right)\right.
\end{multline}
Ribosomes are  a costly resource in the cell,  a common biological assumption is that their number is smaller than the maximum number of places on mRNAs that can accommodate ribosomes. This suggests that the relation 
\[
\sum_{p=1}^P \sum_{k=1}^{K_p^N} \E(G_{p,k})=\sum_{p=1}^P K_p^N \E(G_{p,1}) > N
\]
holds.  i.e. that
\[
\sum_{p=1}^P K_p^N \rho_p\frac{C_p\rho_p^{C_p+1}-(C_p+1)\rho_p^{C_p}+1}{(1-\rho_p)(1-\rho_p^{C_p+1})}>N
\]
is satisfied. In view of the scaling condition~\eqref{beta}, this gives a definition of saturation in the asymptotic regime. 
\begin{definition}
The saturation condition holds if the following relation is satisfied
\begin{equation}\label{sat2}
\sum_{p=1}^P \beta_p\rho_p\frac{C_p\rho_p^{C_p+1}-(C_p+1)\rho_p^{C_p}+1}{(1-\rho_p)(1-\rho_p^{C_p+1})}>1.
\end{equation}
\end{definition}

\subsection{The System under the Saturation Condition}
The following theorem is the key result of this paper, the asymptotic distribution of the number of free ribosomes at equilibrium is obtained in this limiting regime. 
\begin{theorem}\label{theo1}
Under the limiting regime~\eqref{beta} with the saturation condition~\eqref{sat2}, as $N$ goes to infinity, at equilibrium the number of free ribosomes converges in distribution to a Poisson random variable with parameter $\delta(\underline{C})$ where $\delta(\underline{C})<1$ is the unique non-negative solution $y$ of the equation 
\begin{equation}\label{gamma}
\sum_{p=1}^P \beta_p \rho_py\frac{C_p\rho_p^{C_p+1}y^{C_p+1}-(C_p+1)\rho_p^{C_p}y^{C_p}+1}{(1-\rho_py)(1-\rho_p^{C_p+1}y^{C_p+1})}=1,
\end{equation}
$\underline{C}=(C_p)$ and $\rho_p=\lambda_p/\mu_p$. 
\end{theorem}
The basic ingredients of the proof are a change of probability distribution and a  strengthened version of the central limit theorem for lattice distributions, known as Edgeworth expansion.  The method turns out to be a probabilistic analogue of a saddle point method applied to the series
\[
\sum_{\substack{x=(x_{p,k})\in{\cal S}\\r(x)=n}}\; \prod_{p=1}^P\prod_{k=1}^{K_p^N}  \rho_p^{x_{p,k} }
\]
See Flajolet and Sedgewick~\cite[chapter VIII]{Flajolet} for example.  A related result, although of quite different nature,  holds in congested communication networks. The analogue of Equation~\eqref{gamma} is the {\em Erlang Fixed Point Equation}, see Kelly~\cite{Kelly:14}. It is obtained through an optimization formulation which is in fact a saddle point method.   One of the main advantages of the probabilistic formulation is that the technical framework is significantly lighter. 

For $N\geq 1$, with the same notations as before for the random variables $(G_{p,k})$, denote by
\[
S_N=N-\sum_{p=1}^P\sum_{k=1}^{K_p^N}G_{p,k}.
\]
From Equation~\eqref{Free} of Corollary~\ref{propR}, the distribution of $R_N$ the number of free ribosomes can be expressed as 
\begin{equation}\label{eqFree}
\P(R_N=\ell) =\frac{1}{Z_{R_N}}\frac{1}{\ell!} \P(S_N=\ell),
\end{equation}
for $\ell\geq 0$ where $Z_{R_N}$ is the normalization constant. 

To derive the asymptotic behavior of the sequence $(\P(S_N=\ell))$ one introduces, for $\theta>0$, the random variable  ${S}^\theta_N$ whose distribution is defined by, for a non-negative function $f$ on $\N$, 
\[
\E\left(f\left({S}^{\theta}_N\right)\right)=
\frac{1}{\phi_N(\theta)}\E\left(f\left(S_N\right)e^{\theta S_N}\right),
\]
with $\phi_N(\theta)=\E\left(\exp(\theta S_N)\right)$. For $u\in(0,1)$ 
\[
\E\left(u^{{S}^{\theta}_N}\right)= \frac{1}{\phi_N(\theta)}\E\left(\left(ue^\theta\right)^{S_N}\right)=\frac{u^Ne^{N\theta}}{\phi_N(\theta)}\prod_{p=1}^P\left(\frac{1-(u^{-1}\rho_pe^{-\theta})^{C_p+1}}{1-u^{-1}\rho_pe^{-\theta}}\right)^{K_p^N}.
\]
In particular ${S}^\theta$ can be represented as a sum of independent random variables
\[
{S}^\theta_N=N-\sum_{p=1}^P\sum_{k=1}^{K_p^N}G^\theta_{p,k},
\]
where, for $1\leq p\leq P$,  $(G^{\theta}_{p,k},k\geq 0)$ is a sequence of geometric random variables  on $\{0,1,\ldots,C_p\}$ with parameter $\rho_p e^{-\theta}$.

By definition of $S^\theta_N$, for $n\geq 0$,
\begin{multline}\label{eqdim}
\P(S_N\geq \ell)=\phi_N(\theta)\E\left(\ind{S^\theta_N\geq \ell}e^{-\theta S^\theta_N}\right)\\=
\phi_N(\theta)\int_0^{+\infty} \theta e^{-\theta u} \P\left(\ell\leq S^\theta_N\leq u\right)\,\diff u.
\end{multline}
To study  the asymptotic behavior of the distribution of $R_N$,  Equation~\eqref{eqFree} shows that it is enough to investigate the asymptotic distribution of $S_N$. As it will be seen, one will  choose a convenient $\theta$ for which the asymptotic distribution of $S_N^\theta$ can be estimated and  Relation~\eqref{eqdim}  will give the corresponding result for the asymptotic distribution of $S_N$. 

\begin{lemma}\label{lemma1}
For   $N\geq 1$ sufficiently large, there exists a unique $\theta_N$ such that $\E(S^{\theta_N}_N)=0$. The sequence $(\theta_N)$ is converging to $\theta_\infty=-\log\delta(\underline{C})$, where $\delta(\underline{C})$ is the unique non-negative solution of Equation~\eqref{gamma}, furthermore
\[
\E((S_N^{\theta_N})^2)= \sigma^2 N+O\left(\sqrt{N}\right) \text{ with } 
\sigma_\infty=\sqrt{\sum_{p=1}^P \beta_p \var(H_p)},
\]
where, for $1\leq p\leq P$,  $H_p$ is a random variable with  a geometric distribution with parameter $y_p$ which is  truncated at $C_p$ with $y_p=\rho_p\delta(\underline{C})$. 
\end{lemma}
\begin{proof}
The function $\theta\mapsto \phi_N(\theta)$ is strictly convex with $\phi_N(0)=1$ and
\begin{multline*}
\phi_N'(\theta)=\E(S_N^\theta)=N-\sum_{p=1}^PK_p^N\E(G_{p,1}^\theta)\\= N\left(1-
\sum_{p=1}^P \beta_p \rho_pe^{-\theta}\frac{C_p\rho_p^{C_p+1}e^{-\theta(C_p+1)}-(C_p+1)\rho_p^{C_p}e^{-\theta C_p}+1}{(1-\rho_pe^{-\theta})(1-\rho_p^{C_p+1}e^{-\theta(C_p+1)})}\right)+o(\sqrt{N}).
\end{multline*}
The saturation condition~\eqref{sat2} shows that for $N$ sufficiently large, $\phi_N'(0)<0$,  since $\phi_N(\theta)$ converges to infinity when $\theta$ gets large. As a consequence, there exists a unique $\theta_N>0$ such that $\phi'(\theta_N)=0$, i.e. such that $\E(S^{\theta_N}_N)=0$ which gives
\[
N-\sum_{p=1}^P K_p^N \rho_pe^{-\theta_N}\frac{C_p\rho_p^{C_p+1}e^{-\theta_N(C_p+1)}-(C_p+1)\rho_p^{C_p}e^{-\theta_NC_p}+1}{(1-\rho_pe^{-\theta_N})(1-\rho_p^{C_p+1}e^{-\theta_N(C_p+1)})}=0.
\]
Hence, with the above expansion, one gets
\[
1-\sum_{p=1}^P\beta_p \rho_pe^{-\theta_N}\frac{C_p\rho_p^{C_p+1}e^{-\theta_N(C_p+1)}-(C_p+1)\rho_p^{C_p}e^{-\theta_NC_p}+1}{(1-\rho_pe^{-\theta_N})(1-\rho_p^{C_p+1}e^{-\theta_N(C_p+1)})}=o(1/\sqrt{N}),
\]
this shows that, for $N$ sufficiently large, $\exp(-\theta_N)$ is in any arbitrary small neighborhood of $\delta(C)$. This gives the desired convergence of $(\theta_N)$. The other expansion is done with direct, straightforward, calculations. 
\end{proof}

\begin{proof}[Proof of Theorem~\ref{theo1}]
The strategy of the proof is as follows, one derives a uniform asymptotic expansion as $N$ goes to infinity of 
\begin{equation}\label{rap}
\frac{\P(S_N\geq \ell)}{ \phi_N(\theta)}=\int_0^{+\infty} \theta e^{-\theta u} \P\left(\ell \leq S^\theta_N\leq u\right)\,\diff u, \quad \forall \ell\geq 0,
\end{equation}
and, therefore, of ${\P(S_N=\ell)}/{\P(S_N\geq 0)}$. Then an asymptotic expansion of the normalization constant $Z_{R_N}$  of Representation~\eqref{eqFree} is obtained and therefore of the distribution of $R_N$. 

Let $\sigma_N=\sqrt{{\var(S_N^\theta)}/{N}}$, by Lemma~\ref{lemma1} the sequence $(\sigma_N)$ is converging to $\sigma_\infty$. Relation~\eqref{BoundApp} gives the bound 
\begin{multline}\label{Ams}
\sup_{ 0\leq v\leq u} \left|\P\left(v\leq S^{\theta_N}_N\leq u\right)-\frac{1}{\sqrt{2\pi }}
\int_{v/(\sigma_N\sqrt{N})}^{u/(\sigma_N\sqrt{N})} e^{-v^2/2}\,\diff v \right. \\\left.
+ \frac{1}{\sqrt{2\pi}}\left( R_N\left(u/(\sigma_N\sqrt{N})\right)e^{-u^2/(N\sigma_N^2)}- R_N\left(v/(\sigma_N\sqrt{N})\right)e^{-v^2/(N\sigma_N^2)}\right)\right|
\\= o\left(\frac{1}{\sqrt{N}}\right)
\end{multline}

By plugging this {\em uniform} estimation in the integral of the right hand side of Equation~\eqref{eqdim} with $\theta=\theta_N$,
one gets that
\begin{align}
\frac{\P(S_N> \ell)}{ \phi_N(\theta)}&=\int_0^{+\infty} \theta_N e^{-\theta_N u} \P\left(\ell < S^{\theta_N}_N\leq u\right)\,\diff u\notag \\
&= \frac{1}{\sqrt{2\pi}}\int_{0}^{+\infty}
  \int^{u/(\sigma_N\sqrt{N})}_{\ell /(\sigma_N\sqrt{N})}e^{-v^2/2}\,\diff v \,\theta_N e^{-\theta_N u} \,\diff u +o\left(\frac{1}{\sqrt{N}}\right)\notag \\
&= \frac{e^{-\theta_N\ell}}{\sigma_N\sqrt{2\pi N}}\int_{0}^{+\infty}
  \int^{u+\ell}_{\ell}e^{-v^2/\left(2\sigma_N^2 N\right)}\,\diff v \,\theta_N e^{-\theta_N u} \,\diff u +o\left(\frac{1}{\sqrt{N}}\right)
\notag \\ &=\frac{1}{\theta_\infty\sigma_\infty\sqrt{2\pi N}}e^{-\theta_\infty \ell}+o\left(\frac{1}{\sqrt{N}}\right),\label{Berry}
\end{align}
 {\em uniformly} with respect to $\ell \geq 0$.
Indeed,  for $i=1$, $2$, the  polynomials $Q_{in}$ in the definition of $R_N$, see Proposition~\ref{Pipi}, have bounded degree and bounded coefficients and therefore that $x\mapsto Q_{in}(x)\exp(-x^2)$ is a bounded function in $n$ and $x\geq 0$. Hence the term
\[
R_N\left(u/(\sigma_N\sqrt{N})\right)e^{-u^2/(N\sigma_N^2)}
\]
is uniformly  $o(1/\sqrt{N})$ because for $i=1$, $2$,  the  coefficients of the $Q_{iN}(\cdot)$ and of the periodic functions  $S_i(\cdot)$   are  combinations of $L_{\nu N}$, $\nu=3$, $4$  and $1/B_{2 N}^\nu$, $\nu=1/2$, $1$  which  are $o(1/\sqrt{N})$   except for the term $L_{3N}Q_{1N}$, but for this last term the fact that $Q_{1N}(0)=0$ gives nevertheless the appropriate estimate.  

Equation~\eqref{eqdim} gives therefore that
\[
\lim_{N\to+\infty} \sup_{\ell\geq 0} \left|\frac{\P(S_N\geq \ell)}{\P(S_N\geq 0)}-e^{-\ell\theta_\infty}\right|=0
\]
and consequently
\[
\lim_{N\to+\infty} \sup_{\ell\geq 1} \left|\frac{\P(S_N=\ell)}{\P(S_N\geq 0)}-e^{-\ell\theta_\infty}\left(1-e^{-\theta_\infty}\right)\right|=0.
\]
This uniform asymptotic result can then be plugged into the expression~\eqref{eqFree} of the distribution of the number of free ribosomes, by using the fact that the normalization constant $Z_{R_N}$ is 
\[
\sum_{k} \frac{1}{k!} \P(S_N=k).
\]
The limit result of Lemma~\ref{lemma1} concludes the proof of  the theorem.
\end{proof}

The fixed point equation~\eqref{gamma} has a simple intuitive explanation. First one notices that the number of free ribosomes evolves on a very rapid time scale of the order of $N$, since the number of attached ribosomes and the number of free places for ribosomes on mRNAs are of this order.  The number of free ribosomes is therefore quickly at equilibrium, let $\delta(\underline{C})$ be its average . As an approximation, for a given class $(p,k)$ of proteins, provided it is less than $C_p$ and positive,  the number of ribosomes attached to an mRNA increases by $1$ at rate $\delta(\underline{C})\lambda_p$ and decreases by $1$ at rate $\mu_p$.  At equilibrium,  one gets that the average number of ribosomes attached is
\[
\E(X_{p,k})=\rho_p\delta(\underline{C})\frac{C_p\rho_p^{C_p+1}\delta(\underline{C})^{C_p+1}-(C_p+1)\rho_p^{C_p}\delta(\underline{C})^{C_p}+1}{(1-\rho_p\delta(\underline{C}))(1-\rho_p^{C_p+1}\delta(\underline{C})^{C_p+1})},
\]
and the average total number of attached ribosomes is  then
\[
\sum_{p=1}^{P} K_p^N \E(X_{p,k}).
\]
This quantity must be of the order of $N$ (few free ribosomes)  which gives the desired fixed point equation for $\delta(\underline{C})$. 

\medskip

We now look at more detailed characteristics of the protein production process, namely the number of ribosomes attached to mRNAs of class $(p,k)$, $1\leq p\leq P$, $1\leq k\leq K_p^N$ at equilibrium. 
For $1\leq p\leq P$, since all the proteins with a fixed concentration have the same parameters, at equilibrium the random variables $X_{p,k}$, $1\leq k\leq K_p^N$ have the same distribution.  The following theorem gives a limit result for the distribution of these variables as $N$ gets large. 

\begin{theorem}\label{TheoEqui}
Under the limiting regime with the saturation condition~\eqref{sat2}, as $N$ goes to infinity,  at equilibrium 
\[
\lim_{N\to+\infty} (X^N_{p,1}, 1\leq p\leq P)= \left(Y_p, 1\leq p\leq P\right)
\]
for the convergence in distribution, where $Y_p$, $1\leq p\leq P$ are independent random variables and  $Y_p$ has   a truncated  geometric distribution   on $\{0,1,\ldots,C_p\}$ with parameter $\rho_p\delta(\underline{C})$ where $\delta(\underline{C})$ is the unique non-negative solution of Equation~\eqref{gamma}.
\end{theorem}
\begin{proof}
The notations of  the previous proof will be used. 
From Equation~\eqref{pi} for the equilibrium distribution,   one gets that for any non-negative function $f$ on $\N^P$,
\begin{equation}\label{Job}
\E\left(f\left(X^N_{p,1}, 1\leq p\leq P\right)\right)=\frac{1}{Z_N}
\sum_{n=0}^{+\infty}\frac{1}{n!}
\E\left(f\left(G_{1,1},\ldots,G_{p,1}\right)\ind{r((G_{p,k}))=n}\right),
\end{equation}
where, for $1\leq p\leq P)$ and $1\leq k\leq K_p^N$,  $G_{p,k}$  is a geometric random variable with parameter $\rho_p$ restricted to the state space $\{0,1,\ldots, C_p\}$,
and $Z_N$ is the normalization constant.  If $(m_p)\in\N^P$, it has been seen in the previous proof that
\begin{multline*}
\P((G_{p,1})=(m_p), r((G_{p,k}))=N)\\=\phi_N(\theta_N)\E\left(\ind{(G^{\theta_N}_{p,1})=(m_p), r((G^{\theta_N}_{p,k}))=n} e^{-\theta_N {S}^{\theta_N}_N}\right),
\end{multline*}
where, as before,  the random variables $(G^{\theta_N}_{p,k})$ are  truncated geometric random variables with parameter $\rho_p\exp(-\theta_N)$, similarly for $S_N^{\theta_N}$.
Consequently
\begin{multline*}
\E\left(\ind{(G^{\theta_N}_{p,1})=(m_p), r((G^{\theta_N}_{p,k}))=n} e^{-\theta_N S_N^{\theta_N}}\right)\\=
\P\left((Y_p)=(m_p)\right)\E\left(\ind{r((\overline{G}_{p,k}))=n+\|m\|} e^{-\theta_N (\overline{S}_N-\|m\|}\right),
\end{multline*}
where $\|m\|=m_1+\cdots+m_P$ and the bar $\overline{\ }$ indicates that all the components with an index $k=1$ in $(G_{p,k}^{\theta_N})$ are removed. 
From Estimation~\eqref{Berry}, one gets that
\[
\E\left(\ind{r((\overline{G}_{p,k}))=n+\|m\|} e^{-\theta_N (\overline{S}_N-\|m\|}\right)=
\frac{1}{\theta_N\sigma_N\sqrt{2\pi N}}e^{-\theta_N n}(1-e^{-\theta_N})+o\left(\frac{1}{\sqrt{N}}\right),
\]
holds uniformly in $n$ and $m$. By plugging this relation in Equation~\eqref{Job} where $f=\ind{m}$, one gets that
\[
\lim_{N\to+\infty} \P((X^N_{p,1})=(m_p))=\P\left((Y_p)=(m_p)\right). 
\]
The theorem is proved. 
\end{proof}
The above theorem is in fact a mean-field result. In the limit any finite subset of  components of $(X^N_{p,k})$  is converging in distribution to independent random variables. 
\section{Analysis of the Fixed Point Equation}\label{BioSec}
This section is devoted  to the insight one can get on the consequences of Theorem~\ref{TheoEqui} in a biological context.  As it can be seen, the key parameter $\delta(\underline{C})$, the solution  of the  fixed point equation~\eqref{gamma} is expressed in terms of  quite complicated expressions involving the vectors of the loads $(\rho_p)$ and of the capacities $(C_p)$. A simple and tight approximation  of $\delta(\underline{C})$  is  presented. It is shown in the following that  for a wide range of values for the $\rho_p$'s, the solution  $\delta(\underline{C})$ depends in fact in a simple way on the vector of capacities  $\underline{C}=(C_p)$. 

The basic ingredient of the estimations of  $\delta(\underline{C})$ is the simple following fact. If $G$ is a random variable whose distribution is a  truncated geometric distribution with parameter $\rho$ on $\{0,1,\ldots, C\}$, if $\rho<1$ then, provided that $C$ is sufficiently large,  $G$ is close in distribution to a plain geometric distribution with parameter $\rho$, in particular its expected value is close to $\rho/(1-\rho)$. In the case where $\rho>1$, then $G$ can be expressed as $C-G_0$, where $G_0$  is a  truncated geometric distribution with parameter $1/\rho<1$.  Explicit bounds on the accuracy of the approximations  of $\delta(\underline{C})$ are provided. 

\subsection{The underloaded case} 
In this section, it will be assumed that all the loads $\rho_p$, $1\leq p\leq P$, are smaller than $1$. The main result is that the quantity $\delta(\underline{C})$ is essentially independent of $\underline{C}$. 
\begin{lemma}
Under the saturation condition~\eqref{sat2} and if  $\rho_p<1$,  for all   $1\leq p\leq P$ 
and if $\delta_0$ is the unique solution of the equation
\begin{equation}\label{appgamma}
\sum_{p=1}^P \beta_p \frac{\rho_p\delta_0}{1-\rho_p\delta_0}=1,
\end{equation}
then the solution  $\delta(\underline{C})$  of  the fixed point Equation~\eqref{gamma} is such that
$\delta_0\leq \delta(\underline{C})$ and 
\begin{equation}\label{approx}
\left|\delta(\underline{C})-\delta_0\right|\leq  \frac{1}{\sum_{p=1}^P \beta_p\rho_p}\sum_{p=1}^P \beta_p\frac{(C_p+1)\rho_p^{C_p+1}}{1-\rho_p^{C_p}}
\end{equation}
with $\underline{C}=(C_p)$. 
\end{lemma}
\begin{proof}
For $0\leq y<1$ and $c>0$,  one has
\begin{equation}\label{ineqaux}
\frac{cy^{c+1}-(c+1)y^{c}+1}{1-y^{c+1}}=1-(c+1)y^c\frac{1-y}{1-y^c}\leq 1.
\end{equation}
The saturation condition~\eqref{sat2} and the fact that all $\rho_p$ are less than $1$ imply that
\[
\sum_{p=1}^{P} \beta_p\frac{\rho_p}{1-\rho_p}>1. 
\]
Consequently, for $0\leq x\leq 1$,  there exists a unique solution $\phi(x)<1$ of the equation 
\[
\sum_{p=1}^P \frac{\beta_p\rho_p\phi(x)}{1-\phi(x)\rho_p}=x.
\]
Note that $\delta_0$ is simply $\phi(1)$. It is not difficult to see that $\phi$ is differentiable and 
\[
\phi'(x)\sum_{p=1}^P \frac{\beta_p\rho_p}{(1-\phi(x)\rho_p)^2}=1,
\]
and therefore that $0\leq \phi'(x)\leq 1/(\beta_1\rho_1+\beta_2\rho_2+\cdots+\beta_P\rho_P)$. 

By using again Relation~\eqref{ineqaux}, Equation~\eqref{gamma} can then be rewritten as
\[
\delta(\underline{C})=\phi\left(1+\sum_{p=1}^P \beta_p(C_p+1)\frac{(\rho_p\delta(\underline{C}))^{C_p+1}}{1-(\rho_p\delta(\underline{C}))^{C_p}}\right),
\]
consequently, one obtains the relation
\[
\left|\delta(\underline{C})-\delta_0\right|\leq \sum_{p=1}^P \beta_p(C_p+1)\frac{(\rho_p\delta(\underline{C}))^{C_p+1}}{1-(\rho_p\delta(\underline{C}))^{C_p}}\left/\sum_{p=1}^P \beta_p\rho_p\right.
\]
and therefore Relation~\eqref{approx}. 
\end{proof}
The above lemma states that if the right hand side of Relation~\eqref{approx} is sufficiently small then the simple constant $\delta_0$ is an accurate estimation of the parameter $\delta(\underline{C})$. 
Since the $\rho_p$'s are $<1$,  this approximation will hold in the case where the values of the $C_p$'s are not too small which is a reasonable biological assumption.  One gets therefore the following proposition. 
\begin{proposition}[System with only  Underloaded  Classes of mRNAs]
Under the saturation condition~\eqref{sat2} and if  $\rho_p<1$,  for all   $1\leq p\leq P$, 
then, under the limiting regime~\eqref{beta}, at equilibrium  as $N$ goes to infinity, 
\begin{enumerate}
\item the distribution of the number of free ribosomes is Poisson with parameter
$\delta_0+o(h(\underline{C}))$,
\item the distribution of the  number of ribosomes attached to an mRNA of class $(p,k)$ is geometric with parameter $\rho_p\delta_0+o(h(\underline{C}))$.
\end{enumerate}
where $\delta_0$ is the unique solution of the equation 
\[
\sum_{p=1}^P \beta_p \frac{\rho_p\delta_0}{1-\rho_p\delta_0}=1 \text{ and } h(\underline{C})= \sum_{p=1}^P \beta_p\frac{(C_p+1)\rho_p^{C_p+1}}{1-\rho_p^{C_p}}. 
\]
\end{proposition}
\noindent
As a consequence, under the assumptions of the above proposition, one gets that the production rate of proteins of class $(p,k)$ is given by 
\[
\frac{\lambda_p\delta_0}{1-\rho_p\delta_0}+o(h(\underline{C})).
\]

\subsection{The  Case of Overloaded mRNAs} 
The case where at least one of the $\rho_p$ is strictly greater than $1$ is investigated. 
Without loss of generality one can assume the ordering $\rho_1\leq \rho_2\leq\cdots\leq\rho_P$, now suppose that $\rho_P>1$ holds and let
\[
L=\sup\{p\geq 1: \rho_p<1\},
\]
with the convention that $\sup\emptyset=0$.  For simplicity, it is assumed that none of $\rho_p$ is equal to $1$. 

With the same point of view as in the beginning of this section for the truncated geometric random variables associated to the equilibrium of the system: if $\rho_p<1$, then the truncation of the geometric random variable $G$  at $C_p$ has little impact in the analytic expressions and, conversely, if $\rho_p>1$, then the random variable can be expressed as $C_p-G$ where $G$ has a  geometric distribution with parameter $1/\rho_p$. 

The saturation condition can then be written as 
\begin{equation}\label{sat3}
\sum_{p=1}^{L} \beta_p \frac{\rho_p}{1-\rho_p}+\sum_{p=L+1}^{P} \beta_p \left(C_p-\frac{1}{\rho_p-1}\right)+\phi(\underline{C})>1,
\end{equation}
with 
\[
\phi(\underline{C})= -\sum_{p=1}^L \beta_p\frac{(C_p+1)\rho_p^{C_p+1}}{1-\rho_p^{C_p}}+\sum_{p=L+1}^P \beta_p\frac{(C_p+1)}{\rho_p(\rho_p^{C_p}-1)}.
\]
The following proposition gives a simplified version of the results of Section~\ref{Equi} in this case. 
\begin{proposition}[System with Overloaded Classes of  mRNAs]
If the saturation condition~\eqref{sat2} holds and, for some  $1\leq L<P$,  one has $\rho_1\leq \rho_2\leq\cdots\leq \rho_L<1<\rho_{L+1}\leq\cdots\leq \rho_P$ and
\begin{equation}\label{appsat3}
\sum_{p=1}^{L} \beta_p \frac{\rho_p}{1-\rho_p}+\sum_{p=L+1}^{P} \beta_p \left(C_p-\frac{1}{\rho_p-1}\right)>1
\end{equation}
then, under the limiting regime~\eqref{beta}, at equilibrium  as $N$ goes to infinity, 
\begin{enumerate}
\item the distribution of the number of free ribosomes is Poisson with parameter
$\delta_0+o(h(\underline{C}))$.
\item For $1\leq p\leq P$ and $k\geq 1$ the number $R_{p,k}$ of ribosomes attached to an mRNA of class $(p,k)$  is such that
\begin{itemize}
\item if $p\leq L$\\$R_{p,k}$ is  a geometric r.v.  with parameter  $\rho_p\delta_0+o(h(\underline{C}))$; 
\item if $L+1\leq p\leq P$\\
$C_p-R_{p,k}$ is  a geometric r.v. with parameter $\delta_0/\rho_p+o(h(\underline{C}))$,
\end{itemize}
\end{enumerate}
where $1/\rho_{L+1}\leq\delta_0\leq 1 $ is the unique solution of the equation 
\begin{equation}\label{F2}
\sum_{p=1}^{L} \beta_p \frac{\rho_p\delta_0}{1-\rho_p\delta_0}+\sum_{p=L+1}^{P} \beta_p \left(C_p-\frac{1}{\rho_p\delta_0-1}\right)=1
\end{equation}
and
\[
h(\underline{C})= \sum_{p=1}^L \beta_p\frac{(C_p+1)\rho_p^{C_p+1}}{1-\rho_p^{C_p}}+\sum_{p=L+1}^P \beta_p\frac{(C_p+1)}{\rho_p(\rho_p^{C_p}-1)}
\]
\end{proposition}
\begin{proof}
The existence of a solution $1/\rho_{p}<\delta_0<1$ to Equation~\eqref{F2} is a consequence of Condition~\eqref{appsat3} and therefore that $\delta_0\rho_p>1$ for $L+1\leq p\leq P$.  The rest of the proof follows the same arguments as in the proof of the above lemma. 
\end{proof}
Coming back to the proof of Theorem~\eqref{theo1} where  a change of probability was used, for $1\leq p\leq P$, a geometric random variable with parameter $\rho_p$ is transformed  with a reduced parameter $\delta(\underline{C})\rho_p$, it should be noted that this  change of distribution does not  affect the set of overloaded classes under the assumptions of the above proposition, i.e.  for $L+1\leq p\leq P$, one has $\delta(\underline{C})\rho_p>1$. 
\appendix

\section{Edgeworth Expansion for the distribution of  ${S}^{\theta_N}_N$}
Recall that the variable ${S}^{\theta_N}_N$ is given by 
\[
{S}^{\theta_N}_N=N-\sum_{p=1}^P\sum_{k=1}^{K_p^N}G^{\theta_N}_{p,k},
\]
where, for $1\leq p\leq P$,  $(G^{\theta_N}_{p,k},k\geq 0)$ is a sequence of independent geometric random variables  on $\{0,1,\ldots,C_p\}$ with parameter $\rho_p e^{-\theta_N}$. The quantity $\theta_N$ is such that $\E(S_N^{\theta_N})=0$. 
Define
\begin{equation}\label{nN}
n_N=\sum_{p=1}^P\sum_{k=1}^{K_p^N} \E\left(G^{\theta_N}_{p,k}\right),
\end{equation}
to simplify notations one writes in the following $n$ instead of $n_N$ and sets $W_{n}=S_N^{\theta_N}$. 

The random variables $(\xi_{\ell,n},1\leq \ell\leq n)$ are defined as $\xi_{\ell,n}=G_{p_\ell,k_\ell}$ where, for $1\leq \ell \leq n$, $p_\ell$ and $k_\ell$ are such that
\[
 \sum_{p=1}^{p_{\ell}-1} K_p^N< \ell \leq \sum_{p=1}^{p_{\ell}} K_p^N,\quad  k_\ell=\sum_{p=1}^{p_{\ell}} K_p^N-\ell.
\]
Since $\theta_N$ is such that $\E(W_{n})=0$, $W_{n}$ can be represented as 
\[
W_{n}=\sum_{i=1}^{n} \E(\xi_{i,n})-\xi_{i,n},
\]
i.e. $W_{n}$ is the sum of independent centered truncated geometric random variables.
With these notations, one gets, by Lemma~\ref{lemma1}, 
\begin{equation}\label{n/N}
\lim_{N\to+\infty} \frac{n_N}{N}=\sum_{p=1}^P\beta_p\E\left(G_{p,1}^{\theta_\infty}\right). 
\end{equation}

In this context, with the independence property of  geometric random variables, it is quite natural to have a Central Limit  approximation for $S_N^{\theta_N}$. We need a more refined result, namely an {\em Edgeworth expansion}, of the form 
\begin{equation}\label{Edge1}
\sup_{x\in\R} 
\left| \P\left(\frac{W_{n}}{\sqrt{\E(W_{n}^2)}} \leq x\right) -\Phi(x) -\frac{R(x)}{\sqrt{\E(W_{n}^2)}} \right| =o\left(\frac{1}{\sqrt{\E(W_{n}^2)}}\right), 
\end{equation}
where 
\[
\Phi(x)= \frac{1}{\sqrt{2\pi }} \int_{-\infty}^x e^{-v^2/2}\,\diff v \text{ and } R(x)=Q(x)e^{-x^2/2}.
\]
When  the variables $(\xi_{\ell,n})$ are i.i.d. and with some condition on the corresponding characteristic function of the common distribution, Expansion~\eqref{Edge1} holds  and  $Q$ is a polynomial in this case. See Gnedenko and Kolmogorov~\cite[Theorem~2, page~210]{Gnedenko:01} or  Petrov~\cite[Theorem~1, page~159 and Theorem~7, page~175]{Petrov}. 

Unfortunately,the above condition for the expansion excludes  the case of geometric distributions that we have, in fact it excludes lattice distributions in general. In this case, a similar expansion may  hold but with a more complicated technical framework,
\[
R(x)=\left(Q(x)+L\left(x\sigma_N\sqrt{N}\right)\right)e^{-x^2/2},
\]
where $Q(x)$ is some polynomial  and $L$ is a periodic function. See Gnedenko and Kolmogorov~\cite[Theorem~1, page~213]{Gnedenko:01} or   Petrov~\cite[Theorem~6, page~171]{Petrov}  

Note that the random variables $(\xi_{\ell,n})$  are not identically distributed but still independent, there is also analogous  analogous expansion~\eqref{Edge1} but the level of technicality of the proof is an order of magnitude higher. The approach taken by  Barbour and {\v{C}}ekanavi{\v{c}}ius~\cite{Barbour} is of establishing a bound on the distance between  the distribution of $W_{n}$ and  some shifted Poisson distribution, Stein method is the main ingredient of the proof.  The proof of the  bounds obtained by Pipiras~\cite{Pipiras} follows a more classical approach, it turns out that they are more convenient in our case. 
The following proposition is just a direct application of the expansion obtained by Pipiras~\cite[Theorem~1]{Pipiras}, its main purpose is of gathering the (sometimes quite involved)  definitions and various estimates for our special case. 

\begin{proposition}\label{Pipi}
If $W_{n}=S_N^{\theta_N}$ and $n=n_N$ are defined by Relation~\eqref{nN}, then, for all $N\geq 1$, 
\begin{equation}\label{BoundApp}
\sup_{x\in\R} \left| \P\left(\frac{W_{n}}{\sqrt{\E(W_{n}^2)}} \leq x\right) -\Phi(x) - R_{n}(x)e^{-x^2/2}\right|\leq  C_0\left(L_{4n}+\frac{1}{B_{2n}^{3/2}}\right)
\end{equation}
for some constant $C_0>0$, and $R_n=Q_n+D_n$ with
\begin{align*}
Q_n(x)&= Q_{1n}(x)L_{3n}+Q_{2n}(x)L_{4n}\\
D_n(x)&=\frac{1}{\sqrt{B_{2n}}}S_1\left(x \sqrt{B_{2n}}\right) \cdot e^{x^2/2}\frac{\diff}{\diff x} \left(e^{-x^2/2}Q_n(x)\right)\\
&+ \frac{1}{B_{2n}} S_2\left(x \sqrt{B_{2n}}\right)\cdot e^{x^2/2}\frac{\diff^2}{\diff^2 x} \left(e^{-x^2/2}Q_n(x)\right).
\end{align*}
The functions $Q_{1n}$ and $Q_{2n}$ are polynomials whose degree and coefficients are bounded with respect to $n$ such that $Q_{1n}(0)=0$ and the functions $S_1$ and $S_2$ are bounded. 
The sequences  $(B_{\nu n})$ and $(L_{\nu n})$ satisfy
\[
\lim_{N\to+\infty} \frac{B_{\nu n}}{n}=A_\nu\stackrel{\text{def.}}{=}\left.\sum_{p=1}^P \beta_p\E\left(\left[G_{p,1}^{\theta_\infty}-\E(G_{p,1}^{\theta_\infty})\right]^\nu\right)\right/\sum_{p=1}^P \beta_p
\]
where $\theta_\infty$ is defined in Lemma~\ref{lemma1} and
\begin{equation}\label{apeq1}
\lim_{N\to+\infty} {n^{\nu/2-1}}{L_{\nu n}}=\frac{A_\nu}{A_2^{\nu/2}}.
\end{equation}
\end{proposition}
 For the  general definitions of the polynomials $(Q_{\nu n})$, see Statuljavi{\v{c}}us~\cite[Theorem~4]{Statu} or  Gnedenko and Kolmogorov~\cite[page~191]{Gnedenko:01} . For   $S_1$ and $S_2$, see Theorem~1 of Pipiras~\cite{Pipiras}. 
\begin{proof}
With the notations of Pipiras~\cite{Pipiras}, for $\nu\geq 1$,
\[
B_{\nu n}\stackrel{\text{def.}}{=}\sum_{i=1}^{n} \E\left(\xi_{i,n}^\nu\right), \quad L_{\nu n}\stackrel{\text{def.}}{=}\frac{B_{\nu n}}{B_{2n}^{\nu/2}},
\]
Lemma~\ref{lemma1} gives the corresponding asymptotics for the sequences $(B_{\nu n})$ and $(L_{\nu n})$. 
It is not difficult to check that the conditions (1)-(4) of \cite[Theorem~1]{Pipiras} are satisfied: conditions (1)-(3) because the second moments of the centered truncated geometric random variables are bounded below and their third moment is bounded. Condition~(4) holds because $L_{3n}$ is of the order of $1/\sqrt{n}$.  By taking $r=4$ in this result, one  gets the desired relation.
\end{proof}

\providecommand{\bysame}{\leavevmode\hbox to3em{\hrulefill}\thinspace}
\providecommand{\MR}{\relax\ifhmode\unskip\space\fi MR }
\providecommand{\MRhref}[2]{%
  \href{http://www.ams.org/mathscinet-getitem?mr=#1}{#2}
}
\providecommand{\href}[2]{#2}

\end{document}